\documentclass[preprint]{imsart}
\RequirePackage[OT1]{fontenc}
\RequirePackage{amsthm,amsmath}
\RequirePackage[numbers]{natbib}
\RequirePackage[colorlinks,citecolor=blue,urlcolor=blue]{hyperref}

\startlocaldefs
\numberwithin{equation}{section}
\theoremstyle{plain}

\newtheorem{lem}{Lemma}[section]
\endlocaldefs
\usepackage[ruled,vlined]{algorithm2e}
\usepackage{graphicx} 
\usepackage{afterpage,lscape}
\usepackage{subcaption}

\begin{document}

\begin{frontmatter}
\title{Jointly Sparse Global SIMPLS Regression}
\runtitle{Jointly Sparse Global SIMPLS Regression}

\begin{aug}
\author{\fnms{Tzu-Yu} \snm{Liu}\ead[label=e1]{tyliu@eecs.berkeley.edu}}
\address{University of California, Berkeley, CA, USA.\\
\printead{e1}}

\author{\fnms{Laura} \snm{Trinchera}\ead[label=e2]{laura.trinchera@neoma-bs.fr}}
\address{NEOMA Business School, Mont-Saint-Aignan, France.\\
\printead{e2}}

\author{\fnms{Arthur } \snm{Tenenhaus}\ead[label=e3]{Arthur.Tenenhaus@supelec.fr}}
\address{Sup\'elec, Gif-sur-Yvette, France.\\
\printead{e3}}

\author{\fnms{Dennis} \snm{Wei}\ead[label=e4]{ dwei@us.ibm.com}}
\address{IBM T. J. Watson Research Center, NY, USA.\\
\printead{e4}}

\author{\fnms{Alfred O.} \snm{Hero}\thanksref{t1}\ead[label=e5]{hero@eecs.umich.edu}}
\address{University of Michigan, Ann Arbor, MI, USA.\\
\printead{e5}}

\thankstext{t1}{ to whom correspondence should be addressed.}  
\runauthor{T.-Y. Liu et al.}


\end{aug}

\begin{abstract}
Partial least squares (PLS)  regression  combines dimensionality reduction and prediction using a latent variable model. Since partial least squares regression (PLS-R) does not require matrix inversion or diagonalization,  it can be applied to problems with large numbers of variables. As predictor dimension increases, variable selection becomes essential to avoid over-fitting, to provide more accurate predictors and to yield more interpretable parameters. We propose a global variable selection approach that penalizes the total number of variables across all PLS components. Put another way, the proposed global penalty encourages the selected variables to be shared among the PLS components. We formulate PLS-R with joint sparsity as a variational optimization problem with objective function equal to a novel global SIMPLS criterion plus a mixed norm sparsity penalty on the weight matrix. The mixed norm sparsity penalty is the $\ell_1$ norm of the $\ell_2$ norm on the weights corresponding to the same variable used over all the PLS components. A novel augmented Lagrangian method is proposed to solve the optimization problem and soft thresholding for sparsity occurs naturally as part of the iterative solution. Experiments show that the modified PLS-R attains better or as good performance with many fewer selected predictor variables.
\end{abstract}

\tableofcontents
\end{frontmatter}

\section{Introduction}
With advancing technology comes the need to extract information from increasingly high-dimensional data, whereas the number of samples is often limited. Dimension reduction techniques and models incorporating sparsity become important solution strategies. Partial least squares regression (PLS-R) combines dimensionality reduction and prediction using a latent variable model. It was first developed for regression analysis in chemometrics \citep{Wold1983,Sjostr1983}, and has been successfully applied to many different areas, including sensory science and, more recently, genetics \citep{martens1999validation,LeCRos08, chun2009expression,chung2010sparse,Chun2011}. 

Moreover, PLS-R algorithm is designed precisely to operate with high dimensional data. Since the first proposed algorithm does not require matrix inversion nor diagonalization but deflation to find the latent components, it can be applied to problems with large numbers of variables. The latent components reduce the dimension by constructing linear combinations of the predictors, which successfully solved the collinearity problems in chemometrics \citep{WolRuh84}. However, the linear combinations are built on all the predictors. The resulting PLS-R model tends to overfit when the number of predictors increases for a fixed number of samples. Therefore, variable selection becomes essential for PLS-R in high-dimensional sample-limited problems. It not only avoids over-fitting, but also provides more accurate predictors and yields more interpretable estimates. For this reason sparse PLS-R was developed by H. Chun and S. Keles \citep{Keles2010}.  The sparse PLS-R algorithm  performs variable selection and dimension reduction simultaneously using an $\ell_1$ type variable selection penalty. However, the $\ell_1$ penalty used in \citep{Keles2010} penalizes each variable in each component independently and this can result in different sets of variables being selected for each PLS component  leading to an excessively large number of variables. 

In this work we propose a global criterion for PLS that changes the sequential optimization for a K component model in Statistically Inspired Modification of PLS (SIMPLS) \citep{Jong1993} into a unified optimization formulation, which we refer to as global SIMPLS. This enables us to perform global variable selection, which penalizes the total number of variables across all PLS components. We formulate PLS-R with global sparsity as a variational optimization problem with the objective function equal to the global SIMPLS criterion plus a mixed norm sparsity penalty on the weight matrix. The mixed norm sparsity penalty is the $\ell_1$ norm of the $\ell_2$ norm on the weights corresponding to the same variable used over all the PLS components. The proposed global penalty encourages the selected variables to be shared among all the $K$ PLS components.
 A novel augmented Lagrangian method is proposed to solve the optimization problem, which enables us to obtain the global SIMPLS components and to perform joint variable selection simultaneously. A greedy algorithm is proposed to overcome the computation difficulties in the iterations, and soft thresholding for sparsity occurs naturally as part of the iterative solution. Experiments show that our approach to PLS regression attains better or as good performance (lower mean squared error, MSE) with many fewer selected predictor variables. These experiments include a chemometric data set, and a human viral challenge study dataset, in addition to numerical simulations.

We review the developments in PLS-R for both univariate and multivariate responses in Section \ref{sec:PLS overview}, in which we discuss different objective functions that have been proposed for PLS-R, particularly the Statistically Inspired Modification of PLS (SIMPLS) proposed by de Jong \citep{Jong1993}. In Section \ref{sec:mix norm}, we formulate the Jointly Sparse Global SIMPLS-R by proposing a new criterion that jointly optimizes over $K$ weight vectors (components) and imposing a mixed norm sparsity penalty to select variables jointly. The algorithmic implementation is discussed in Section \ref{sec:algorithm} with simulation experiments presented in Section \ref{sec:experiment}. The proposed Jointly Sparse Global SIMPLS Regression is applied to two applications: (1) Chemometrics in Section \ref{sec:application1} and (2) Predictive health studies in Section \ref{sec:application2}. Section \ref{sec:conclusion} concludes the paper.

\section{Partial Least Squares Regression}
\label{sec:PLS overview}
Partial Least Squares (PLS) methods embrace a suite of data analysis techniques based on algorithms belonging to the PLS family. These algorithms consist of various extensions of the Nonlinear estimation by Iterative PArtial Least Squares (NIPALS) algorithm that was proposed by Herman Wold \citep{Wold66a} as an alternative algorithm for implementing  Principal Component Analysis (PCA) \citep{Hot33}. The NIPALS approach was slightly modified by Svante Wold, and Harald Martens, in order to obtain a regularized component based regression tool, known as PLS Regression (PLS-R) \citep{Wold1983,WolRuh84}. 

Suppose that the data consists of  $n$ samples  of $p$ independent variables $X \in R^{n \times p} $ and $q$ dependent variables (responses) $Y \in R^{n \times q} $. In standard PLS Regression the aim is to define orthogonal latent components in $R^{ p}$, and then use such latent components as predictors for $Y$ in an ordinary least squares framework. The X weights used to compute the latent components can be specified by using iterative algorithms belonging to the NIPALS family or by a sequence of eigen-decompositions. The general underlying model is $X=TP'+E$ and $Y=TQ'+F$, where $T \in R^{n \times K}$ is the latent component matrix, $P \in R^{p \times K}$ and $Q \in R^{q \times K}$ are the loading matrices, $K$ is the number of components, $E$ and $F$ are the residual terms. The latent components in $T = [\begin{array}{*{20}c}   {{\bf t}_1 } & {{\bf t}_2 } & {...} & {{\bf t}_K }  \\\end{array}]$ are linear combinations of the independent variables, hence PLS can be viewed as a dimensional reduction technique, reducing the dimension from $p$ to $K$. The latent components should be orthogonal to each other either by construction as in NIPALS \citep{Wold1983,WolRuh84} or via constrained optimizations as in SIMPLS \citep{Jong1993}. This allows PLS to build a parsimonious model for high dimensional data with collinearity \citep{WolRuh84}.

\subsection{Univariate Response}
We assume, without loss of generality, that all the variables have been centered in a pre-processing step. For univariate $Y$, i.e $q=1$, PLS Regression, also often denoted as PLS1, successively finds $X$ weights $R = [\begin{array}{*{20}c}   {{\bf r}_1 } & {{\bf r}_2 } & {...} & {{\bf r}_K }  \\\end{array}]$ as the solution to the constrained optimization
\begin{equation}
\label{PLS}
    {\bf r}_k   =   \mathop {\arg \max }\limits_{\bf r}  \{{\bf r}'X'_{(k-1)}YY'X_{(k-1)} {\bf r}\}  \ \ s.t. \ \     {\bf r}'{\bf r} = 1.
\end{equation}
where $X_{(k-1)}$ is the matrix of the residuals (i.e., the deflated matrix) from the regression of the $X$-variables on the first $k-1$ latent components, and  $X_{(0)}=X$. 
These weights are then used to find the latent components 
$T = [\begin{array}{*{20}c}   X_{(0)}{{\bf r}_1 } & X_{(1)}{{\bf r}_2 } & {...} & X_{(K-1)}{{\bf r}_K }  \\\end{array}]$.
Such components can be also expressed in terms of original variables (instead of deflated variables), i.e., as $T=XW$, where $W$ is the matrix containing the weights to be applied to the original variables in order to exactly obtain the latent components \citep{Tenenhaus1998}. 

For a fixed number of components, the response variable $Y$ is predicted in an ordinary least squares regression model, where the latent components play the role of the exogenous variables,
\[ \hat Q=\arg \mathop {\min }\limits_Q \{ ||Y - TQ' ||_2 \}  = (T' T)^{ - 1} T' Y. \] This provides the regression coefficients $\hat \beta ^{PLS}  = W\hat Q'$ for the model $Y=X{\bf \beta}^{PLS}+F$.

Depending on the number of selected latent components the length $\|\hat{\beta}^{PLS}\|_2$ of the vector of PLS coefficients changes. In particular, de Jong \citep{DeJ95} had shown that the sequence of these coefficient vectors has lengths that are  strictly increasing as  the number of components increases. This sequence converges to the ordinary least squares coefficient vector  and the maximum number of latent components obtainable equals the rank of the $X$ matrix. Thus, by using a number of latent components $K<rank(X)$, PLS-R performs a dimension reduction by shrinking the $X$ matrix. Hence, PLS-R is a suitable tool for problems for which the data contains many more variables $p$ than observations $n$.

The objective function in (\ref{PLS}) can be interpreted as maximizing the squared covariance between $Y$ and the latent component: $corr^2 (Y,X_{k-1}{\bf r}_k){\mathop{\rm var}} (X_{k-1}{\bf r}_k)$. Because the response $Y$ has been taken into account to formulate the latent matrix, PLS usually has better performance in prediction problems than principle component analysis (PCA) does. This is one of the main differences between PLS-R and PCA \citep{Boulesteix2006}.

\subsection{Multivariate Response}
Similarly to univariate response PLS-R, multivariate response PLS-R selects latent components in $R^p$ and $R^q$ , i.e., ${\bf t}_k$ and ${\bf v}_k$, such that the covariance between ${\bf t}_k$ and ${\bf v}_k$ is maximized.  For a specific component, the sets of weights ${\bf r}_k \in R^p$ and ${\bf c}_k \in R^q$  are obtained by solving 
\begin{eqnarray}
& &\max \{ {\bf t}' {\bf v}\}  = \max \{ {\bf r}' X'_{(k-1)} Y_{(k-1)}^{}{\bf c}\} \nonumber\\
& & s.t. \ \ {\bf r}'{\bf r}={\bf c}'{\bf c}=1
\label{PLS2crit}
\end{eqnarray}
where ${\bf t}_k=X_{(k-1)}{\bf r}_k$, ${\bf v}_k=Y_{(k-1)}{\bf c}_k$, and  $X_{(k-1)}$ and $Y_{(k-1)}$ are the deflated matrices associated with $X$ and $Y$. Notice that the optimal solution ${\bf c}_k$ should be proportional to $Y'_{k-1}X_{k-1}^{}{\bf r}_k$. Therefore, the optimization in (\ref{PLS2crit}) is equivalent to 
\begin{eqnarray}
   &  & \mathop { \max }\limits_{\bf r} \ \{{\bf r}'{X'_{k-1}Y_{k-1}^{}Y'_{k-1}X_{k-1}^{}} {\bf r}\}  \nonumber\\
& & s.t.   \ \    {\bf r}'{\bf r} = 1.
\label{PLS2}
\end{eqnarray}
For each component, the solution to this criterion can be obtained by using a so called PLS2 algorithm. A detailed description of the iterative algorithm as presented by H\"oskuldsson \citep{Hoskuldsson1988} is in Algorithm \ref{NIPALS} .   
 
\begin{algorithm}
{
\For{ k=1:K}{
initialize ${\bf r}$\\
$X=X_{new}$\\
$Y=Y_{new}$\\
\While{solution has not converged}{
${\bf t}=X{\bf r}$\\
${\bf c}=Y'{\bf t}$\\
Scale {\bf c} to length 1\\
${\bf v}=Y{\bf c}$\\
${\bf r}=X'{\bf v}$\\
Scale {\bf r} to length 1\\
}
loading vector ${\bf p}=X'{\bf t}/({\bf t}'{\bf t})$\\
deflate $X_{new}=X-{\bf t}{\bf p}'$\\
regression ${\bf b}=Y'{\bf t}/({\bf t}'{\bf t})$\\
deflate $Y_{new}=Y-{\bf t}{\bf b}'$\\
${\bf r_k}={\bf r}$\\
}
}

\caption{ PLS2 algorithm \label{NIPALS}}
\end{algorithm}

In 1993 de Jong proposed a variant of the PLS2 algorithm, called Statistically Inspired Modification of PLS (SIMPLS), which calculates the PLS latent components directly as linear combinations of the original variables \citep{Jong1993}.
The SIMPLS was first developed as the solution to an optimization problem
\begin{eqnarray}
 \label{SIMPLS}
& &     {\bf w}_k   = \mathop {\arg \max }\limits_{\bf w}  ({\bf w}'{X'YY'X} {\bf w}) \ \\
& & s.t. \ \ {\bf w}'{\bf w}  = 1,  \ \ {\bf  w}' {X'X}{\bf w}_j  = 0 \  for \  j=1,...,k-1. \nonumber 
 \end{eqnarray}
Ter Braak and de Jong \citep{Braak2006} provided a detailed comparison between the objective functions for PLS2 in (\ref{PLS2}) and SIMPLS in (\ref{SIMPLS}) and showed that the successive weight vectors  ${\bf w}_k$ can be derived either from the deflated data matrices or the original variables in PLS2 and SIMPLS respectively. Let $W^{+}$ be the Moore-Penrose inverse of $W=[{{\bf w}_1}\ {{\bf w}_2}\ ... \ {{\bf w}_{k-1}}]$. The PLS2 algorithm (Algorithm \ref{NIPALS}) is equivalent to solving the optimization 
\[
    {\bf w}_k  = \mathop {\arg \max }\limits_{\bf w} ({\bf w}'{X'YY'X} {\bf w}) \]
\[ \ s.t. \
    {\bf w}'(I-WW^{+}){\bf w} = 1, \\ {\bf  w}' {X'X}{\bf w}_i  = 0 \  for \  i=1,...,k-1.
\]
Both NIPALS and SIMPLS have the same objective function but each is maximized under  a different normalization constraint. NIPALS and SIMPLS  are equivalent when Y is univariate, but provide slightly different weight vectors in multivariate scenarios. The performance depends on the nature of the data, but SIMPLS  appears easier to interpret since it does not involve deflation of the data set \citep{Jong1993}. We develop our globally sparse PLS-R based on the SIMPLS optimization formulation.


\section{Sparse Partial Least Squares Regression}
\label{sec:mix norm}

\subsection{$\ell_1$ Penalized Sparse PLS Regression}
One approach to sparse PLS-R is to add the $\ell_1 $ norm of the weight vector, a sparsity inducing penalty, to (\ref{SIMPLS}). The solution for the first component would be obtained by solving
\begin{eqnarray}
\label{PLS_keles}
& &     {\bf w}_1   = \mathop {\arg \max }\limits_{\bf w}  ({\bf w}'{X'YY'X} {\bf w}) \\ \nonumber
& &      s.t. \ \ {\bf w}'{\bf w}  = 1,   \ ||{\bf w}|| _1\le \lambda. 
\end{eqnarray}
The addition of the $\ell_1$ norm is  similar to SCOTLASS (simplified component lasso technique), the sparse PCA proposed by Jolliffe \citep{Jolliffe2003}. However, the solution of SCOTLASS is not sufficiently sparse, and the same issue remains in (\ref{PLS_keles}). Chun and Keles \citep{Keles2010} reformulated the problem, promoting the exact zero property by imposing the $\ell_1$ penalty on a surrogate of the weight vector instead of the original weight vector \citep{Keles2010}. For the first component, they solve the following optimization by alternating between  updating ${\bf w}$ and updating  ${\bf z}$  (block coordinate descent). 
\[
{\bf w}_1,{\bf z}_1=
\mathop {\arg \min }\limits_{{\bf w},{\bf z}} \{  - \kappa {\bf w}' X' YY' X{\bf w} + (1 - \kappa )({\bf z} -{\bf w})' X' YY' X({\bf z} - {\bf w}) + \lambda _1 ||{\bf z}||_1  + \lambda _2 ||{\bf z}||_2^2 \}\ \] \[ s.t. \ {\bf w}' {\bf w} = 1
\]
Allen {\it et. al} proposed a general framework for regularized PLS-R \cite{Allen2013}. 
\[
\mathop {\max }\limits_{{\bf w},{\bf v}} \ {\bf w}'M{\bf v} - \lambda P({\bf w})\ \ s.t. \ \ {\bf w}'{\bf w} \le 1,{\rm  }{\bf v}'{\bf v} = 1
\]
in which $M$ is the cross-product matrix $X'Y$, and the regularization function $P$ is a convex penalty function. The formulation is a relaxation of SIMPLS with penalties being applied to the weight vectors, and can be viewed as a generalization of \citep{Keles2010}.

As mentioned in the Introduction, these formulations (\cite{Keles2010, Allen2013})  penalize the variables in each PLS component independently. This paper proposes an alternative in which variables are penalized simultaneously over all components.  First, we define the global weight matrix, consisting of the $K$ weight vectors, as

\[
W = \left[ {\begin{array}{*{20}c}
   {\begin{array}{*{20}c}
   |  \\
   {{\bf w}_1 }  \\
   |  \\
\end{array}} & {\begin{array}{*{20}c}
   |  \\
   {{\bf w}_2 }  \\
   |  \\
\end{array}} & {\begin{array}{*{20}c}
   {}  \\
    \cdots   \\
   {}  \\
\end{array}} & {\begin{array}{*{20}c}
   |  \\
   {{\bf w}_K }  \\
   |  \\
\end{array}}  \\
\end{array}} \right] = \left[ {\begin{array}{*{20}c}
   {\begin{array}{*{20}c}
    -  & {{\bf w}'_{(1)} } &  -   \\
\end{array}}  \\
   {\begin{array}{*{20}c}
    -  & {{\bf w}'_{(2)} } &  -   \\
\end{array}}  \\
   {\begin{array}{*{20}c}
   {} &  \vdots  & {}  \\
\end{array}}  \\
   {\begin{array}{*{20}c}
    -  & {{\bf w}'_{(p)} } &  -   \\
\end{array}}  \\
\end{array}} \right].
\]
Notice that  the elements in a particular row of W, i.e., ${\bf w}'_{(j)}$, are all associated with the same predictor variable ${\bf x}_j$. Therefore, rows of zeros correspond to variables that are not selected. To illustrate the drawbacks of penalizing each variable in each component independently, as in \citep{Keles2010},  suppose that each entry in $W$ is selected independently with probability $p_1$. The probability that the $(j)_{th}$ variable is not selected becomes $(1-p_1)^K$, and the probability that all the variables are selected by at least one weight vector  is $[1-(1-p_1)^K]^p$, which increases as the number of weight vectors $K$ increases. This suggests that for large $K$ the local variable selection approach of  \citep{Keles2010} may not lead to an overall sparse and parsimonious PLS-R model. In such cases a group sparsity constraint can be employed to limit the overall number of selected variables. 

\subsection{Jointly Sparse Global SIMPLS Regression}
The Jointly Sparse Global SIMPLS Regression variable selection problem is to find the top $K$ weight vectors that best relate $X$ to $Y$, while using limited number of variables. This is a subset selection problem that is equivalent to adding a constraint on the $\ell_0$ norm of the vector consisting of  the norms of the rows of $W$, i.e,  the number of nonzero rows in $W$.   For concreteness we use the $\ell_2$ norm for the rows. This leads to the optimization problem

\begin{eqnarray}
\label{global_PLS}
& & W  = \mathop{\arg \min}\limits_W  - \frac{1}{{n^2 }}\sum\limits_{k = 1}^K {\bf w}'_kX'YY'X {\bf w}_k \\ \nonumber
& & \ \ s.t.  \ 
  { {||{\boldsymbol \varpi } ||_0 }\le t },   \  {{\bf w}'_k{\bf w}_k = 1 \ \forall \ k},\ and \ {\bf  w}'_k X'X {\bf w}_i  = 0 \ \forall \ \ i \ne k
\end{eqnarray}
in which 
\[
 {\boldsymbol \varpi } = \left[ {\begin{array}{*{20}c}
   {||{\bf w}_{(1)} ||_2 }  \\
   {||{\bf w}_{(2)} ||_2 }  \\
    \vdots   \\
   {||{\bf w}_{(p)} ||_2 }  \\
\end{array}} \right].
\]

The objective function (\ref{global_PLS}), which we refer to as global SIMPLS, is the sum of the objective functions  (\ref{SIMPLS}) in the first $K$ iterations of SIMPLS. Instead of the sequential greedy solution in PLS2 algorithm, the proposed jointly sparse global SIMPLS Regression solves for the $K$ weight vectors simultaneously. We introduce the $\frac{1}{n^2}$ factor to the objective function to interpret it as an empirical covariance. Given the complexity of this combinatorial problem,  as in standard optimization practice, we relax the $\ell_0$ norm optimization to a mixed norm structured sparsity penalty \citep{Bach2008}. 
\begin{eqnarray}
\label{global_PLS_relax}
W = \mathop{\arg \min}\limits_W  - \frac{1}{{n^2 }}\sum\limits_{k = 1}^K {\bf w}'_k X'YY'X{\bf w}_k +\lambda\sum\limits_{j = 1}^p {||{\bf w}_{(j)} ||_2 }   \\  \nonumber
s.t. \ \ {\bf w}'_k{\bf w}_k = 1 \ \forall\ k\ and \ {\bf  w}'_k X'X {\bf w}_i  = 0 \ \forall  \ i \ne k
\end{eqnarray}
The $\ell_2$ norm of each row of $W$ promotes grouping entries in $W$ that relate to the same predictor variable, whereas the $\ell_1$ norm promotes a small number of groups, as in (\ref{PLS_keles}).


\section{Algorithmic Implementation for Jointly Sparse Global SIMPLS  Regression}
\label{sec:algorithm}
Constrained eigen-decomposition and group variable selection are each well-studied problems  for which efficient algorithms have been developed.  We propose to solve the optimization (\ref{global_PLS_relax}) by augmented Lagrangian methods, which allows one to solve (\ref{global_PLS_relax}) by variable splitting iterations. Augmented Lagrangian methods  introduce a new variable $M$, constrained such that $M=W$, such that the row vectors ${\bf m}_{(j)}$ of $M$ obey the same structural pattern as the rows of $W$:
\begin{eqnarray}
\label{global_PLS_AL}
& & \mathop{\min}\limits_{W,M}   - \frac{1}{{n^2 }}\sum\limits_{k = 1}^K {\bf w}'_k X'YY'X{\bf w}_k +\lambda\sum\limits_{j = 1}^p {||{\bf m}_{(j)} ||_2 }   \nonumber\\ 
& & \ \ s.t.   \  {\bf w}'_k{\bf w}_k   = 1 \ \forall \ k\ ,  \ {\bf  w}'_k X'X{\bf w}_i  = 0 \  \forall \ i \ne k, \ and \ M=W
\end{eqnarray}
The optimization (\ref{global_PLS_AL}) can be solved by replacing  the constrained problem by an unconstrained one with an additional penalty on the Frobenius norm of the difference $M-W$. This penalized optimization can be iteratively solved by an alternating direction method of multipliers (ADMM) algorithm \cite{Eckstein1992,Goldstein2009,Afonso2011,Boyd2011,Hong2012,Goldstein2012,Ramani2012,Nien2014},  a block coordinate descent method that  alternates between optimizing over $W$ and over $M$ (See Algorithm \ref{PLS_algorithm}). We initialize Algorithm \ref{PLS_algorithm} with $M^{(0)}$ equal to the solution of SIMPLS, and  $D^{(0)}$ equal to the zero matrix. Setting the parameter $\mu$ is nontrivial \cite{Ramani2012}, and is hand-tuned for fastest convergence in some applications \cite{Afonso2011}. Once the algorithm converges, the final PLS regression coefficients are obtained by applying SIMPLS regression on the selected variables keeping the same number of components $K$. The optimization over $W$ can be further simplified to a secular equation problem, whereas the optimization over $M$ can be shown to reduce to a soft thresholding operation. The algorithm iterates until the stopping criterion based on the norm of the residuals $||W^{(\tau)}-M^{(\tau)}||_F<\epsilon$ is satisfied, for some given tolerance $\epsilon$. As described later in the experimental comparisons section, the parameters $\lambda$ and $K$ are decided by cross validation.

\begin{algorithm}
{
set $ \tau=0$, choose $\mu>0$, $M^{(0)}$,  $D^{(0)}$\;
\While{stopping criterion is not satisfied}{
$ W^{(\tau + 1)} = \mathop {\arg \min }\limits_W  - \frac{1}{{n^2 }}\sum\limits_{k = 1}^K {{\bf w}'_k X'YY'X{\bf w}_k  + \frac{\mu}{2}||W - M^{(\tau)} - D^{(\tau)}||_F^2 }  $\\  $ \ \ \ \ \ s.t. 
     \  {\bf w}'_k{\bf w}_k = 1\ \forall k, \  {\bf  w}'_k X'X{\bf w}_i  = 0 \ \forall \  i \ne k$\; 
$ M^{(\tau + 1)} = \mathop {\arg \min }\limits_M \lambda \sum\limits_{j = 1}^p {||{\bf m}_{(j)} ||_2  + } \frac{\mu}{2}||W^{(\tau + 1)} - M - D^{(\tau)}||_F^2  $\;
$ D^{(\tau + 1)} = D^{(\tau)} - W^{(\tau + 1)} + M^{(\tau + 1)} $\;
$\tau=\tau+1$\;
}
}
\caption{{Algorithm for solving the global SIMPLS with global variable selection problem using the augmented Lagrangian method.}\label{PLS_algorithm}}
\end{algorithm}

\paragraph{Optimization over $W$}

The following optimization in Algorithm \ref{PLS_algorithm} is a nonconvex quadratically constrained quadratic program (QCQP). 
\[
 W^{(\tau + 1)} = \mathop {\arg \min }\limits_W  - \frac{1}{{n^2 }}\sum\limits_{k = 1}^K {{\bf w}'_k X'YY'X{\bf w}_k  + \frac{\mu}{2}||W - M^{(\tau)} - D^{(\tau)}||_F^2 }  \]
\[ s.t.      \  {\bf w}'_k{\bf w}_k = 1\ \forall k, \  {\bf  w}'_k X'X{\bf w}_i  = 0 \ \forall \  i \ne k\; 
\]

We propose solving for the $K$ vectors in $W$ successively by a greedy approach. Let ${\bf m}_k$ and ${\bf d}_k$ be the columns of the matrices $M$ and $D$, and ${\boldsymbol \omega}_k={\bf m}_k+{\bf d}_k$. The optimization over $W$ becomes
\begin{eqnarray}
\label{PLS_greedy}
 {\bf w}_k^{(\tau + 1)} & = & \mathop {\arg \min }\limits_{{\bf w}}  - \frac{1}{{n^2 }} {{\bf w}' X'YY'X{\bf w}  + \frac{\mu}{2}||{\bf w}- {\boldsymbol \omega}_k||_2^2 }  \nonumber\\
&  s.t. &  \      \   {\bf w}'{\bf w}   = 1, \  {\bf  w}' X'X{\bf w}_i = 0 \ \forall \  i<k.
\end{eqnarray}


\begin{lem}
Let $N$ be an orthonormal basis for the orthogonal complement of $\{X'X{\bf w}_i\}, i<k$. The optimization (\ref{PLS_greedy}) can be solved by the method of Lagrange multipliers. The solution is ${\bf w}_k  = N(A - \alpha I)^{ - 1} {\bf b}$, in which $A = -  \frac{1}{{n^2 }} N'X' YY' X N$, ${\bf b} = \frac{\mu }{2}N' {\boldsymbol \omega}_k $ and $\alpha$ is the minimum solution that satisfies ${\bf b}'(A-\alpha I)^{-2}{\bf b}=1$.
\end{lem}

\begin{proof}
Let ${\bf w}=N{\bf \tilde w}$, then the optimization (\ref{PLS_greedy}) can be written as 
\[\min \tilde {\bf w}' A\tilde {\bf w} - 2{\bf b}' \tilde {\bf w} \  \ s.t. \ \tilde {\bf w}' \tilde {\bf w} = 1.\]
Since we assume that ${\bf w}$ is a linear combination of the basis vectors in $N$, the orthogonality conditions in (\ref{PLS_greedy}) are automatically satisfied. Hence these conditions have been dropped in the new formulation. Then using Lagrange multipliers, we can show that the solution takes the form as stated above. 

Suppose there are two solutions of $\alpha$ that satisfy  ${\bf b}'(A-\alpha I)^{-2}{\bf b}=1$, corresponding to two pairs of solutions to the optimization, $({\bf \tilde w}_1 ,\alpha _1 )$ and $({\bf \tilde w}_2 ,\alpha _2 )$. Since $\tilde {\bf w}  = (A - \alpha I)^{ - 1} {\bf b}$,
\begin{equation}
\label{min_alpha1}
A{\bf \tilde w}_1=\alpha_1{\bf \tilde w}_1+b
\end{equation}
\begin{equation}
\label{min_alpha2}
A{\bf \tilde w}_2=\alpha_2{\bf \tilde w}_2+b.
\end{equation}
By multiplying (\ref{min_alpha1}) by ${\bf \tilde w}_1$, and (\ref{min_alpha2}) by ${\bf \tilde w}_2$, then subtracting the two new equations, we have  
\begin{equation}
\label{min_alpha3}
{\bf \tilde w}'_1A{\bf \tilde w}_1-{\bf \tilde w}'_2A{\bf \tilde w}_2=\alpha_1-\alpha_2+({\bf \tilde w}'_1-{\bf \tilde w}'_2)b.
\end{equation}
On the other hand, by multiplying (\ref{min_alpha1}) by ${\bf \tilde w}_2$, and (\ref{min_alpha2}) by ${\bf \tilde w}_1$, and subtracting the two new equations, we have  
\begin{equation}
\label{min_alpha4}
({\bf \tilde w}'_1-{\bf \tilde w}'_2)b=(\alpha_1-\alpha_2){\bf \tilde w}'_1{\bf \tilde w}_2.
\end{equation}
Given (\ref{min_alpha3}) and (\ref{min_alpha4}), it can be shown that 
\[
(\tilde {\bf w}'_1 A\tilde {\bf w}_1  - 2{\bf b}' \tilde {\bf w}_1 ) - (\tilde {\bf w}'_2 A\tilde {\bf w}_2  - 2{\bf b}' \tilde {\bf w}_2 ) =
(\alpha_1-\alpha_2)-({\bf \tilde w}'_1-{\bf \tilde w}'_2)b=
 \frac{{\alpha _1  - \alpha _2 }}{2}||\tilde {\bf w}_1  - \tilde {\bf w}_2 ||_2^2.
\]
Hence, one should select the minimum among all the feasible $\alpha$'s. 
\end{proof}

The equation ${\bf b}'(A-\alpha I)^{-2}{\bf b}=1$ is a secular equation, a well studied problem  in constrained eigenvalue decomposition \citep{Gander1989,Beck2006}. The more general problem of least squares with a quadratic constraint was discussed in \citep{Gander1981}. We can diagonalize the matrix $A$ as $A=U DU'$, in which $D$ is diagonal with eigenvalues $d_1, d_2,...,d_p$in decreasing order on the diagonal, and the columns of $U$ are the corresponding eigenvectors. Define

\[
g(\alpha ) = {\bf b}' (A - \alpha I)^{ - 2} {\bf b} = {\bf b}' U \left[ {\begin{array}{*{20}c}
   {\frac{1}{{(d_1  - \alpha )^2 }}} & {} & {} & {}  \\
   {} & {\frac{1}{{(d_2  - \alpha )^2 }}} & {} & {}  \\
   {} & {} &  \ddots  & {}  \\
   {} & {} & {} & {\frac{1}{{(d_p  - \alpha )^2 }}}  \\
\end{array}} \right]U'{\bf b}.
\]
Let ${\bf \tilde b} = U'{\bf b}$, then $g(\alpha ) ={\bf b}'(A-\alpha I)^{-2}{\bf b}= \sum\limits_i^{} {\frac{{\tilde b_i^2 }}{{(d_i  - \alpha )^2 }}} $, and hence $g(\alpha)=1$ is a secular equation. 
 $g(\alpha)$ increases strictly as $\alpha$ increases from $ - \infty$ to $d_p$, since 
\[
g'(\alpha ) = \sum\limits_i^{} {\frac{{2\tilde b_i^2 }}{{(d_i  - \alpha )^3 }}} 
\]
is positive for $- \infty < \alpha < d_p$. Moreover, given the limits 
\[
\mathop {\lim }\limits_{\alpha  \to  - \infty } g(\alpha ) =  0
\]
\[
\mathop {\lim }\limits_{\alpha  \to  d_p^{-} } g(\alpha ) =  \infty
\]
we can conclude that there is exactly one solution $\alpha < d_p$ to the equation $g(\alpha)=1$, \citep{Gander1989}. An iterative algorithm (Algorithm  \ref{SEC} ) is used to solve $g(\alpha)=1$ starting from a point to the left of the smallest eigenvalue $d_p$  \citep{Beck2006}.

Notice that calculating $g(\alpha ) = {\bf b}' (A - \alpha I)^{ - 2} {\bf b}$ involves inverting a $(p-k+1) \times (p-k+1)$ matrix, but $A$ has rank at most $q$. We can reduce the computational burden by the use of Woodbury matrix identity, 
\[(A - \alpha I)^{ - 1}=\frac{-1}{\alpha}(I-\frac{1}{\alpha n^2}N'X'Y(I+\frac{1}{\alpha n^2}Y'XNN'X'Y)^{-1}Y'XN).\]
The new format only requires inverting a $q \times q$ matrix, and in most applications, the number of responses $q$ is much less than the number of predictors $p$. Furthermore, $N$ is involved in $g(\alpha)$ in the form of $NN'=I-HH'$, in which $H$ is an orthonormal basis for $\{X'X{\bf w}_i\}, i<k$. $H$ can be constructed by Gram-Schmidt process as the algorithm successively finds the weight vectors ${\bf w}_i's$.

\begin{algorithm}
{
set $\tau=0$, choose $\alpha _0  = d_p  - \varepsilon _1 $\;
\While{stopping criterion is not satisfied, $|g(\alpha _\tau ) - 1| > \varepsilon _2 $}{
$\alpha _{\tau + 1}  = \alpha _{\tau}  + 2\frac{{g^{ - 1/2} (\alpha _\tau ) - 1}}{{g^{ - 3/2} (\alpha _\tau )g' (\alpha _\tau )}}$\;
}
}
\caption{{Iteration for solving secular equation.}\label{SEC}}
\end{algorithm}

\paragraph{Optimization over $M$}
The optimization over $M$ has a closed form solution. Let $\Delta = W^{(\tau + 1)}  - D^{(\tau)} $, and ${\bf  \delta}_{(j)}$ denote the $j_{th}$ row of $\Delta$, then each row of $M$ is given as ${\bf  m}_{(j)}  = [||{\bf  \delta}_{(j)} || - \frac{\lambda }{\mu}]_ +  \frac{{{\bf  \delta}_{(j)} }}{{||{\bf  \delta}_{(j)} ||}}$, in which $[z]_+=max\{z,0\}$. \\
\\
Convergence analysis of ADMM can be found in \cite{Eckstein1992,Hong2012,Goldstein2012, Nien2014}. In particular, it has been shown that ADMM converges linearly for strongly convex objective functions \cite{Goldstein2012}. Although the convergence is based on strong convexity assumptions, ADMM has been widely applied in practice \cite{Goldstein2009,Afonso2011}, even to nonconvex problems \cite{Boyd2011}.  The proposed Jointly Sparse Global SIMPLS Regression is one example of these nonconvex applications. 

\section{Simulation Experiments}
\label{sec:experiment}

We implement the simulation models in \citep{Bair2006,Keles2010}. There are four models all following  $Y = X\beta  + F $, in which the number of observation is $n=100$, and the dimension is $p=5000$. Full details of the models are given in Table \ref{table:exp_model}. We compare five different methods: the standard PLS regression (denoted as PLS in the following comparison tables), PLS generalized linear regression proposed by Bastien et al. \citep{bastien2005pls}, $\ell_1$ penalized PLS regression  \citep{Keles2010} (denoted as $\ell_1$ SPLS), Lasso \citep{tibshirani1996regression} and the Jointly Sparse Global SIMPLS regression ( denoted as $\ell_1/\ell_2$ SPLS). All the methods select their parameters by ten fold cross-validation, except for the PLS generalized linear regression, which stops including an additional component if the new component is not significant. The parameter $\mu$ in the Jointly Sparse Global SIMPLS-R is fixed to 2000, and updated in each iteration by a scaling factor $1.01\mu$. The experiments on real data in Section \ref{sec:application1} and \ref{sec:application2} are using the same setting of $\mu$. Two i.i.d sets are generated for each trial: one as the training set and one as the test set. Ten trials are conducted for each model, and the averaged results are listed in Table  \ref{table:simulation1}.

In most of the simulations, we observe that the proposed Jointly Sparse Global SIMPLS-R performs as good or better than other methods in terms of the prediction MSE. In particular, the number of variables and the number of components chosen in Jointly Sparse Global SIMPLS-R are usually less than the $\ell_1$ penalized PLS-R. We also calculate the $R^2$ for each method on the training data to measure the variation explained. The standard PLS regression and the PLS generalized linear regression proposed by Bastien et al. \citep{bastien2005pls} both have $R^2$ close to $1$, but the performance in terms of MSE is not ideal in the first three models for these methods. This may suggest that these models overfit the data. In addition to the averaged performance, the p-values of one sided paired t-test suggest that Jointly Sparse Global SIMPLS-R reduces model complexity significantly from those in the standard PLS regression and the PLS generalized linear regression. Lasso achieves low complexity in terms of the number of variables, but the MSE is high compared to Jointly Sparse Global SIMPLS-R and the $\ell_1$ penalized PLS-R. The cross validation time for Jointly Sparse Global SIMPLS Regression is long, searching over a two-dimensional grid of the number of components $K$ and the regularization parameter $\lambda$ to minimize MSE. However, the performance in terms of prediction MSE improves, and the model complexity in terms of the number of variables and the number of components both decreases compared with other methods in most simulations. 

\begin{table}
\scriptsize
\caption{Simulation models in \citep{Bair2006,Keles2010}: These models were originally proposed to test supervised principal components. Model 1 and Model 2 are designed such that one latent variable dominates the multi-collinearity. Model 3 has multiple latent variables, and Model 4 has a correlation structure from an autoregressive process. In these model descriptions, $I$ is used to denote the indicator function. We use $i$ ($1 \le i \le n$) to index the $n$ samples, $j$ ($1 \le j \le p$) to index the $p$ variables, and $k$ to index the hidden components.
\label{table:exp_model}}
\begin{tabular}{  *{1}{c} } 
\hline 
\\{\bf Model 1}\\
$
\begin{array}{l}
 H_{1}(i)  = 3I(i \le 50) + 4I(i > 50), \ 1 \le i \le n   \\ 
 H_{2}(i)  = 3.5 \\ 
 X_j  = H_k  + \varepsilon _j , \ p_{k - 1}  < j \le p_k ,\ k = 1,2, \ (p_0 ,...,p_2 ) = (0,50,p) \hspace{20 mm} \\ 
 \beta_j  = \left\{ {\begin{array}{*{20}c}
   {\frac{1}{{25}}} & {1 \le j \le 50}  \\
   0 & {51 \le j \le p}  \\
\end{array}} \right. \\ 
 \varepsilon _j \ is \ { N (0,I_n )} \ distributed, \ and  \  F \ is \ {N (0,1.5^2 I_n ) } \   distributed.\\ 
 \end{array}
 $
\\\\\hline
\\{\bf Model 2}\\
$
\begin{array}{l}
 H_{1}(i)  = 2.5I(i \le 50) + 4I(i > 50) \\ 
 H_{2}(i)  = 3.5 + 1.5I(u_{1i}  \le 0.4) \\ 
 H_{3}(i)  = 3.5 + 0.5I(u_{2i}  \le 0.7) \\ 
 H_{4}(i)  = 3.5 - 1.5I(u_{3i}  \le 0.3) \\ 
 H_{5}(i)  = 3.5 \\ 
 u_{1i} ,u_{2i} ,u_{3i} \ are\ i.i.d \ from\ Unif(0,1) \\ 
 X_j  = H_k  + \varepsilon _j , \ p_{k - 1}  < j \le p_k , \ k = 1,...,5,(p_0 ,...,p_5 ) = (0,50,100,200,300,p) \\ 
 \beta_j  = \left\{ {\begin{array}{*{20}c}
   {\frac{1}{{25}}} & {1 \le j \le 50}  \\
   0 & {51 \le j \le p}  \\
\end{array}} \right. \\ 
 \varepsilon _j \ is \ { N (0,I_n )} \ distributed, \ and  \  F \ is \ {N (0, I_n ) } \   distributed.\\ 
 \end{array}
$
\\ \\\hline
\\{\bf Model 3}\\
$
\begin{array}{l}
 H_{1}(i)  = 2.5I(i \le 50) + 4I(i > 50) \\ 
 H_{2}(i)  = 2.5I(1 \le i \le 25, \ or \ 51 \le i \le 75) + 4I(26 \le i \le 50, \ or \ 76 \le i \le 100) \\ 
 H_{3}(i)  = 3.5 + 1.5I(u_{1i}  \le 0.4) \\ 
 H_{4}(i)  = 3.5 + 0.5I(u_{2i}  \le 0.7) \\ 
 H_{5}(i)  = 3.5 - 1.5I(u_{3i}  \le 0.3) \\ 
 H_{6}(i)  = 3.5 \\ 
 u_{1i} ,u_{2i} ,u_{3i} \ are\ i.i.d \ from\ Unif(0,1) \\ 
 X_j  = H_k  + \varepsilon _j , \ p_{k - 1}  < j \le p_k , \ k = 1,...,5,(p_0 ,...,p_6 ) = (0,25,50,100,200,300,p) \\ 
 \beta_j  = \left\{ {\begin{array}{*{20}c}
   {\frac{1}{{25}}} & {1 \le j \le 50}  \\
   0 & {51 \le j \le p}  \\
\end{array}} \right. \\ 
 \varepsilon _j \ is \ { N (0,I_n )} \ distributed, \ and  \  F \ is \ {N (0, I_n ) } \   distributed.\\ 
 \end{array}
$
\\\\\hline
\\{\bf Model 4}\\
$
\begin{array}{l}
 H_{1}(i)  = I(i \le 50) + 6I(i > 50) \\ 
 H_{2}(i)  = 3.5 + 1.5I(u_{1i}  \le 0.4) \\ 
 H_{3}(i)  = 3.5 + 0.5I(u_{2i}  \le 0.7) \\ 
 H_{4}(i)  = 3.5 - 1.5I(u_{3i}  \le 0.3) \\ 
 H_{5}(i)  = 3.5 \\ 
 u_{1i} ,u_{2i} ,u_{3i} \ are\ i.i.d \ from\ Unif(0,1) \\ 
X=(X^{(1)},X^{(2)})\\
X^{(1)}\ is \ generated \ from \ N(0,\Sigma_{50 \times 50}),\ \Sigma \  is \  from  \ AR(1)  \ with  \ \rho=0.9.\\
 X_j^{(2)}  = H_k  + \varepsilon _j , \ p_{k - 1}  < i \le p_k ,\ k = 1,...,5,(p_0 ,...,p_5 ) = (0,50,100,200,300,p-50) \\ 
\beta_i=r_m \  for\ p_{m-1} < i \le p_m, \ m=1,...,6, \ where \\
 (p_0,...,p_6)=(0,10,20,30,40,50,p), \ (r_1,...,r_6)=(8,6,4,2,1,0)/25 \\
 \varepsilon _i \ is \ { N (0,I_n )} \ distributed, \ and  \  F \ is \ {N (0,1.5^2 I_n ) } \  distributed.\\ 
 \end{array}
$
\\\\\hline
\end{tabular}
\end{table}

\begin{landscape}
\begin{table}
\caption{Performance comparison table for the 4 simulation models. We compare five different methods: the standard PLS regression (PLS), PLS generalized linear regression proposed by Bastien et al. (Bastien), $\ell_1$ penalized PLS regression ($\ell_1$ SPLS), Lasso  and the Jointly Sparse Global SIMPLS regression ($\ell_1/\ell_2$ SPLS). These results are obtained by averaging over 10 trails. \label{table:simulation1}}
\tiny
\begin{center}

    \begin{tabular}{ *{10}{c}}
\hline \\
& & & & & & \multicolumn{4}{c}{p-values of one sided paired t-test}\\ 
\cline{7-10}\\
  & 1. PLS  & 2.  Bastien  & 3. $\ell_1$ SPLS & 4.  Lasso  & 5. $\ell_1/\ell_2$ SPLS & (5,1) & (5,2) & (5,3) & (5,4) \\ \\\hline   
\\{\bf Model 1}  \\   
number of comp.  & 1.4 & 5 & 1.9 &  NA  & 1.4 & $ 0.5$ &  $ 4.49\times10^{-7}$ &  $ 0.19$ & NA \\     
number of variables  & 5000 & 1129.4 & 246.5 & 40.7 & 276.1 &  $ 1.84\times10^{-11}$ &  $5.46\times10^{-5}$ &  $0.35$ &  $0.053$\\     
MSE  & 3.14 & 2.98 & 3.00 & 3.23 & 2.82 &  $0.01$ &  $0.043$ &  $0.075$ &  $0.006$\\ 
$R^2$  & 0.98 & 1 & 0.71 & 0.59 & 0.83\\        
Time CV  & 101.47 & 0 & 43.49 & 51.59 & 11414\\     
Time analysis  & 0.89 & 121.91 & 0.05 & 0.04 & 6.40\\     
Time prediction  & 0.010 & 0.011 & 0.002 & 0.04 & 0.002    \\     
Total time  & 102.37 & 121.92 & 43.54 & 51.67 & 11420\\  \\   \hline
\\{\bf Model 2} \\
number of comp.  & 2 & 5 & 2.3 &  NA  & 1.1 &  $0.0671$ &  $1.19\times10^{-11}$ &  $ 0.0184$ & NA \\  
number of variables  & 5000 & 1158.4 & 273.4 & 15.8 & 171.7 &  $1.13\times10^{-14}$ &  $ 6.91\times10^{-9}$ &  $0.1041$  &  $0.0119$ \\  
MSE  & 3.18 & 2.99 & 2.93 & 3.09 & 2.69 &  $ 8.69\times10^{-4}$ &  $0.0046$ &  $0.0344$  &  $ 2.70\times10^{-4}$\\  
$R^2$  & 0.98 & 1 & 0.79 & 0.39 & 0.75\\ 
Time CV  & 100.51 & 0 & 43.03 & 53.97 & 11420\\  
Time analysis  & 1.28 & 122.69 & 0.06 & 0.04 & 5.72\\ 
Time prediction  & 0.010 & 0.011 & 0.002 & 0.039 & 0.001\\ 
Total time  & 101.80 & 122.70 & 43.09 & 54.05 & 11426\\  \\\hline
\\{\bf Model 3}\\
number of comp.  & 1.4 & 5 & 1.4 &  NA  & 1.5 &  $0.4201$  & $5.53\times10^{-6}$ &  $ 0.3632$  & NA\\  
number of variables  & 5000 & 1156.4 & 89.2 & 41.3 & 60.5 &  $2.32\times10^{-19}$ &  $ 1.39\times10^{-12}$ &  $0.1697$ &  $ 0.1726$ \\  
MSE  & 1.82 & 1.48 & 1.27 & 1.48 & 1.25 &  $ 1.18\times10^{-4}$ &  $ 0.0104$ &  $ 0.3430$ &  $ 0.0054$ \\  
$R^2$  & 0.98 & 1 & 0.77 & 0.75 & 0.73\\ 
Time CV  & 102.61 & 0 & 43.84 & 49.45 & 11295\\  
Time analysis  & 1.03 & 126.08 & 0.04 & 0.04 & 5.48\\  
Time prediction  & 0.01 & 0.01 & 0.001 & 0.039 & 0.001\\  
Total time  & 103.65 & 126.09 & 43.88 & 49.53 & 11300\\\\  \hline
\\{\bf Model 4}\\    
number of comp.  & 2 & 5 & 2.6 &  NA  & 2.1 & $0.4057$  &  $ 6.82\times10^{-5}$ & $0.2201$  & NA\\  
number of variables  & 5000 & 1118.8 & 1260.8 & 9.4 & 1180.5 &  $9.62\times10^{-5}$ &  $0.4618$ &  $0.3918$  & $0.0485$  \\  
MSE  & 2.15 & 2.29 & 2.41 & 2.14 & 2.36 &  $0.0087$ & $0.1874$ &   $0.3812$ &   $0.0056$ \\  
$R^2$  & 1 & 1 & 0.78 & 0.19 & 0.91\\ 
Time CV  & 98.16& 0 & 44.31 & 50.52 & 12051\\  
Time analysis  & 1.55 & 123.79 & 0.10 & 0.04 & 7.97\\  
Time prediction  & 0.010 & 0.011 & 0.007 & 0.042 & 0.004\\  
Total time  & 99.73 & 123.8 & 44.41 & 50.60 & 12059\\\\  \hline
    \end{tabular}\end{center}

\end{table}
\end{landscape}

\section{Application 1: Chemometrics Study}
\label{sec:application1}

In this section we show experimental results obtained by comparing standard PLS-R, $\ell_1$ penalized PLS-R \citep{Keles2010} (denoted as $\ell_1$ SPLS in the performance table), and our proposed Jointly Sparse Global SIMPLS-R (denoted as $\ell_1/\ell_2$ SPLS  in the performance table). All the methods have been applied on the Octane data set (see \citep{Tenenhaus1998}).  The Octane data is a real data set consisting of 39 gasoline samples for which the digitized Octane spectra have been recorded at 225 wavelengths (in nm). The aim is to predict the Octane number, a key measurement of the physical properties of gasoline, using the spectra as predictors. This is of major interest in real applications, because the conventional procedure to calculate the Octane number is time consuming and involves expensive and maintenance-intensive equipment as well as skilled labor.

The experiments are composed of 150 trials. In each trial we randomly  split the 39 samples into 26 training samples and 13 test samples. The regularization parameter $\lambda$ and number of components $K$ are selected by 2-fold cross validation on the training set. The averaged results over the 150 trials are shown in Table \ref{exp}. We further show the variable selection frequencies for the sparse PLS methods over the 150 trials superimposed on the octane data in Fig. \ref{fig:select_freq_PLS} (B) and (C). In chemometrics, the rule of thumb is to look for variables that have large amplitudes in first derivatives with respect to wavelength. Notice that both $\ell_1$ penalized PLS-R and Jointly Sparse Global SIMPLS-R have selected variables around 1200 and 1350 nm, and the selected region in the latter case is more confined. Box and Whisker plots for comparing the MSE, number of selected variables, and number of components of these three PLS formulations are shown in Fig. \ref{fig:select_freq_PLS} (A).   Comparing our proposed Jointly Sparse Global SIMPLS Regression with standard PLS-R and $\ell_1$ penalized PLS-R \citep{Keles2010}, we show that Jointly Sparse Global SIMPLS-R attains better performance in terms of MSE, the number of predictors, and the number of components. Besides, the model complexity in Jointly Sparse Global SIMPLS-R is significantly lower than both standard PLS-R and $\ell_1$ penalized PLS-R, given the p-values of one sided paired t-test.

\begin{table}
\tiny
\caption{Performance comparison table for the Octane data. \label{exp}}
\begin{center}
    \begin{tabular}{  *{6}{c} }    
    \hline\\
    &  &   &   &  \multicolumn{2}{c}{p-values of one sided paired t-test} \\
\cline{5-6}\\
 & 1. PLS  & 2. $\ell_1$ SPLS  & 3. $\ell_1/\ell_2$ SPLS  & (3,1) & (3,2)\\ \\ \hline
 
number of comp.  & 5.5 & 4.5  &  3.8 & $5.1888\times10^{-16}$ &  0.0027\\ 
number of var.  & 225 & 87.3 & 38.5 &  $1.1176\times10^{-121}$ & $1.0967\times10^{-15}$\\ 
MSE & 0.0564  & 0.0509 & 0.0481 & 0.0032 &  0.1575 \\ \hline
\end{tabular}\end{center}
\end{table}

\begin{figure}
\makebox[12cm][l]{(A)}
\makebox{  
  \includegraphics[height=4cm]{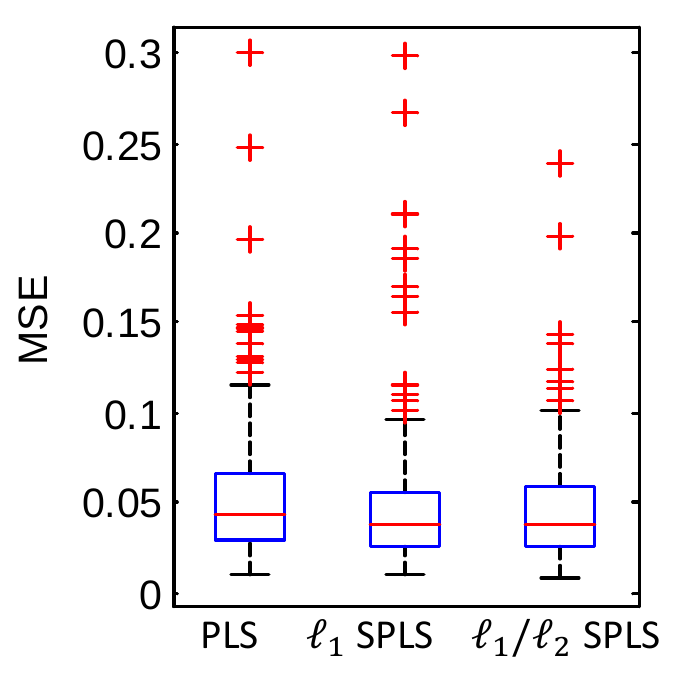}
  \includegraphics[height=4cm]{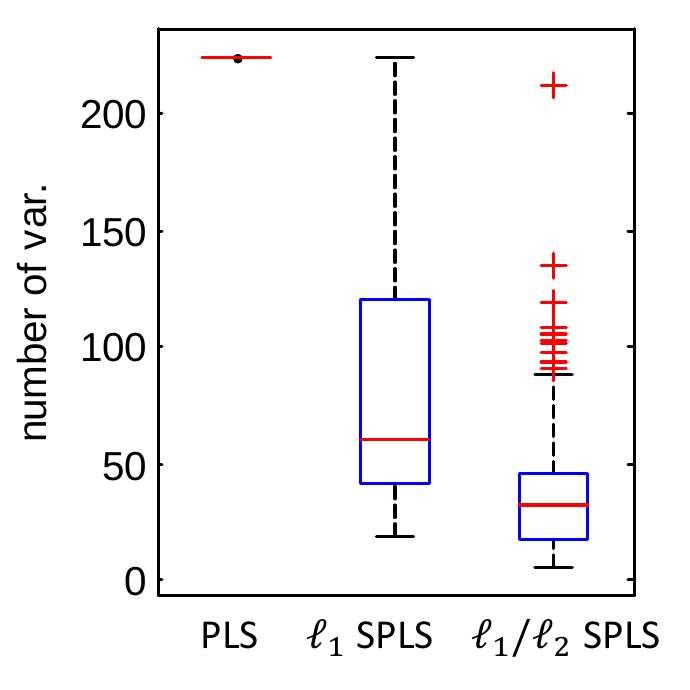}
  \includegraphics[height=4cm]{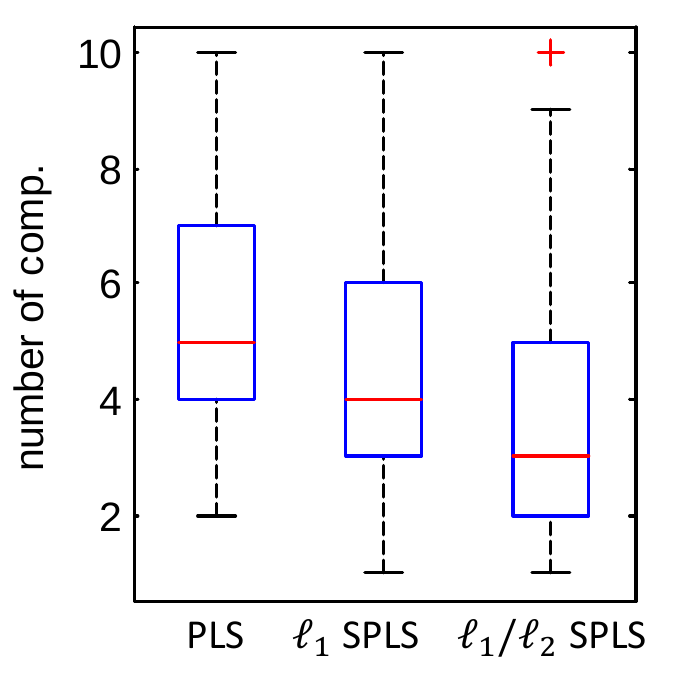}}\\
  \makebox[12cm][l]{} 
\makebox[12cm][l]{(B)}  
\makebox{ 
\includegraphics[height=4.5cm]{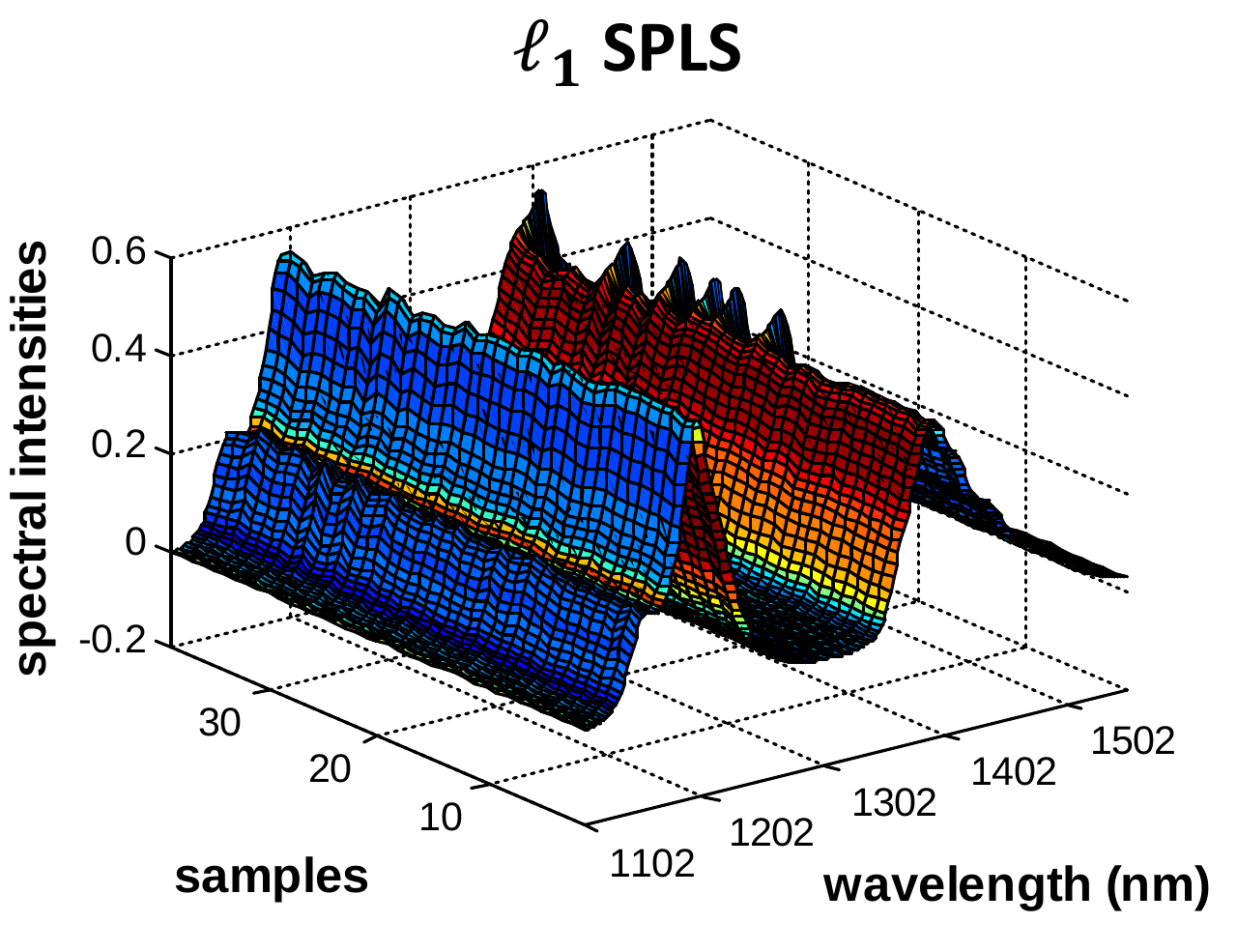} 
\includegraphics[height=4.5cm]{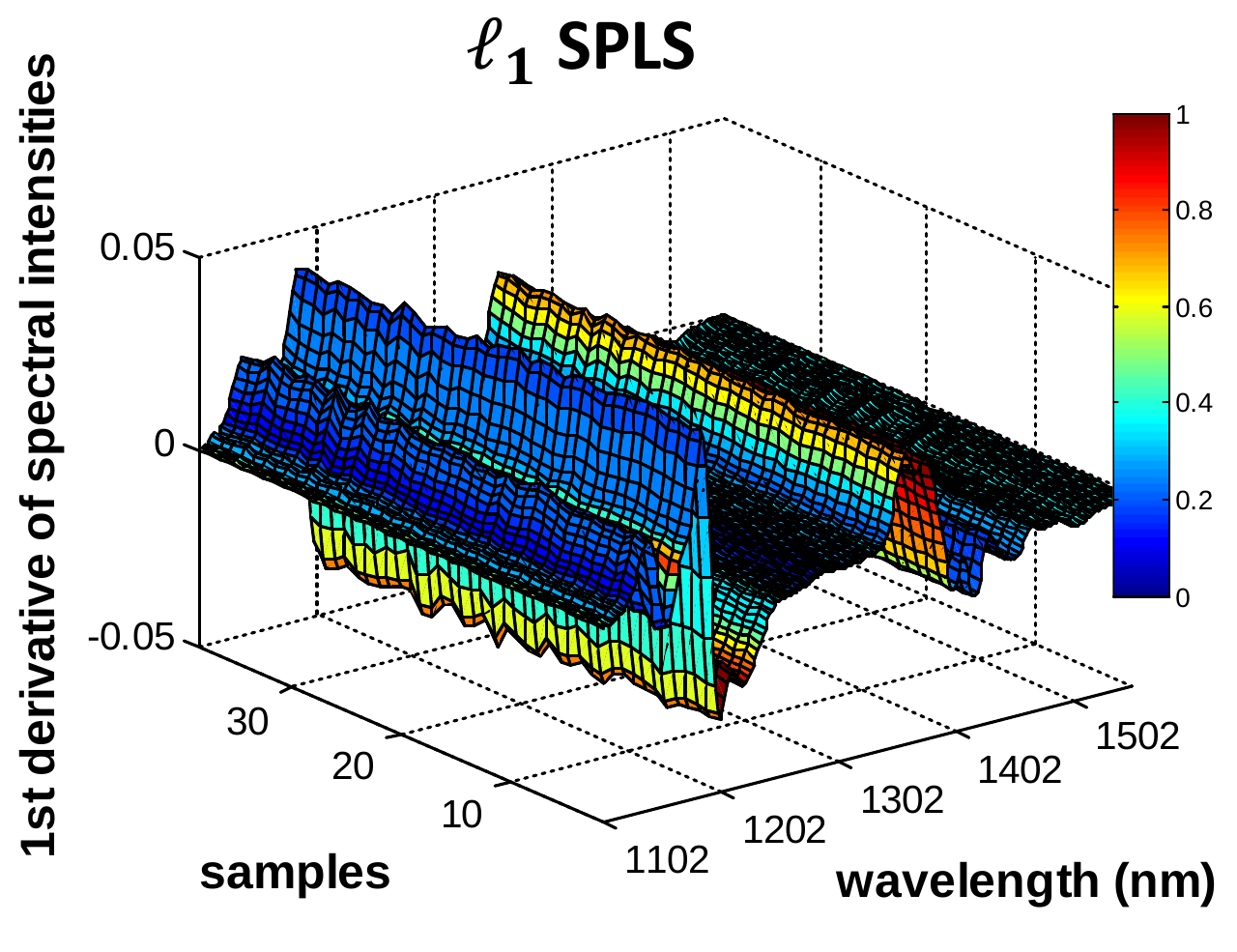}}\\ 
\makebox[12cm][l]{} 
\makebox[12cm][l]{(C)}
\makebox{ 
 \includegraphics[height=4.5cm]{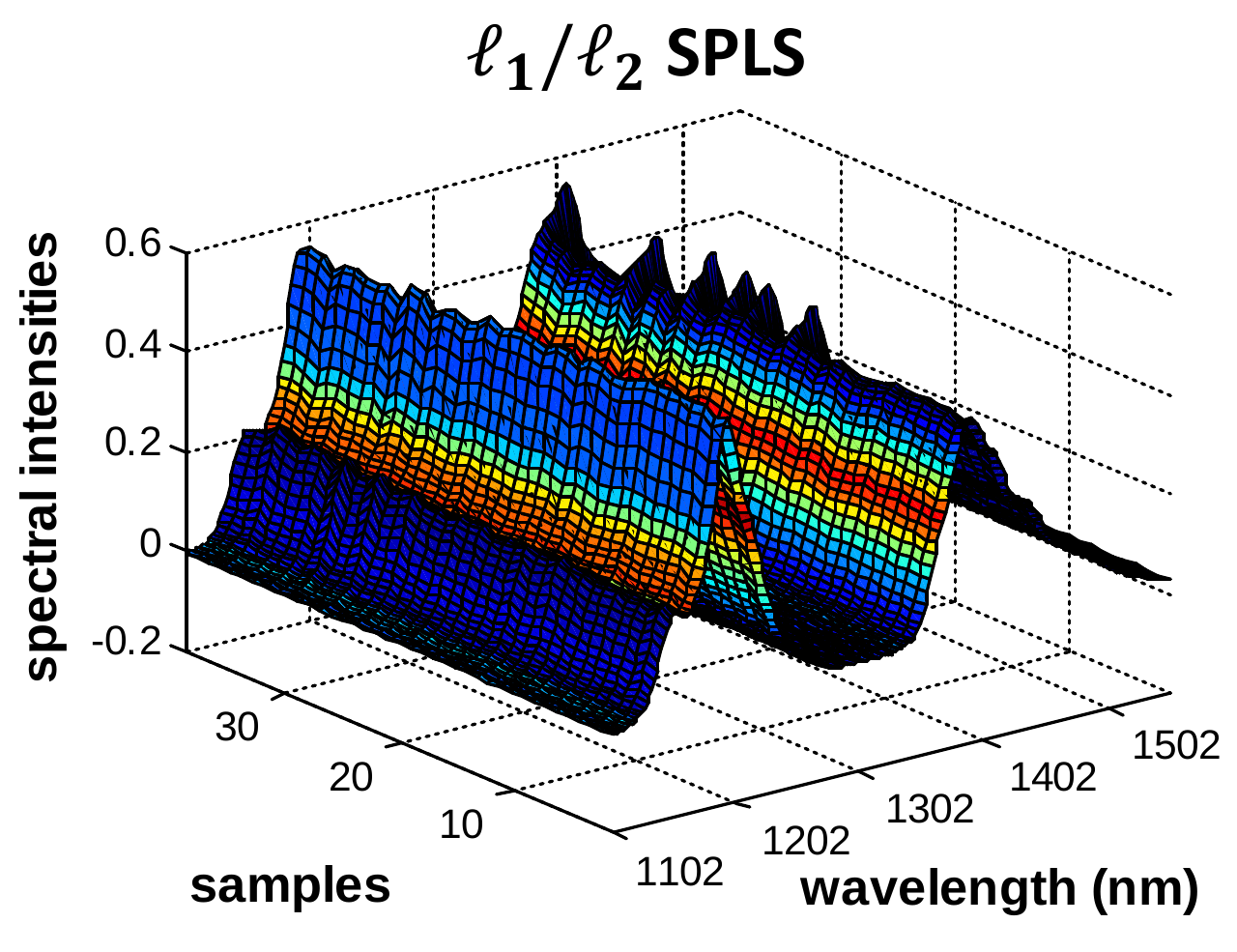} 
 \includegraphics[height=4.5cm]{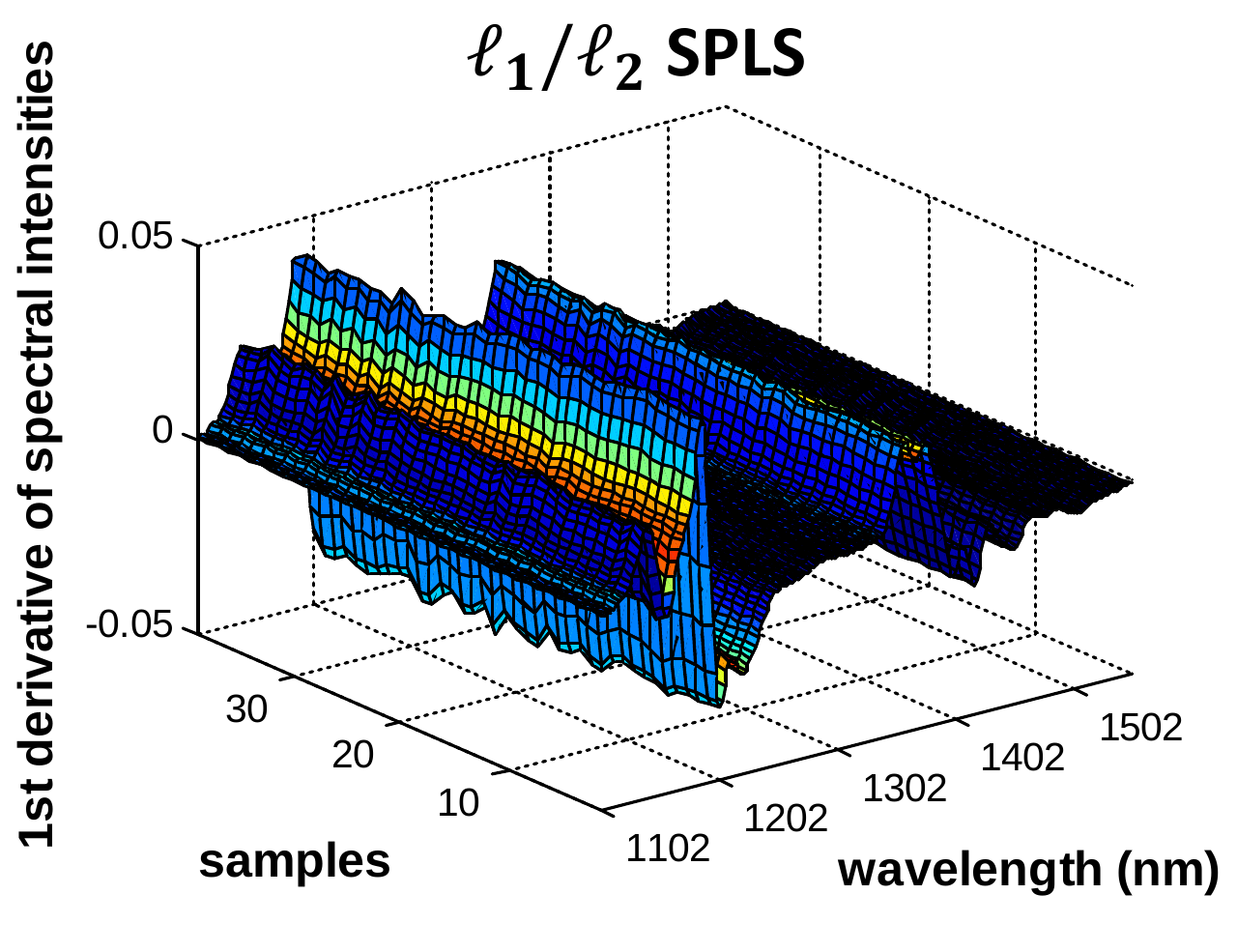}}\\
   \caption{
(A) Box and Whisker plots for comparing the MSE, number of selected variables, and number of components of three PLS formulations: standard PLS-R, $\ell_1$ penalized PLS-R ($\ell_1$ SPLS), our proposed Jointly Sparse Global SIMPLS-R ($\ell_1/\ell_2$ SPLS). (B)  Variable selection frequency of $\ell_1$ SPLS superimposed on the octane data and its first derivative: The height of the surfaces represents the exact value of the data over 225 variables for the 39 samples. The color of the surface shows the selection frequency of the variables as depicted on the colorbar.
(C)Variable selection frequency of $\ell_1/\ell_2$ SPLS superimposed on the octane data and its first derivative.
  \label{fig:select_freq_PLS} }
   \end{figure}

\section{Application 2: Sparse Prediction of Disease Symptoms from Gene Expression}
\label{sec:application2}
In this section we apply the Jointly Sparse Global SIMPLS Regression to 4 types of predictive health challenge studies involving the H3N2, the H1N1, the HRV, and the RSV viruses. In these challenge studies, publicly available from the NCBI-GEO website, serial peripheral blood
samples were acquired from a population of subjects inoculated with live flu viruses \cite{Zaas2009,Huang2011,Woods2013,Zaas2013}. The prediction task in these experiments is to predict the symptom scores based on gene expression of 12023 genes. There were 10 symptom scores, i.e., runny nose, stuffy nose, sneezing, sore throat, earache, malaise, cough, shortness of breath, headache, and myalgia, documented over time. The symptoms are self-reported scores, ranging from 0 to 3. We linearly interpolate the gene expressions to match them with the sampling time of the symptom reports. We compare the Jointly Sparse Global SIMPLS-R with standard PLS-R and $\ell_1$ penalized PLS-R by leaving one subject out as the test set, and the rest as the training set. The process is repeated until all subjects have been treated as the test set. The number of components for all methods and the regularization parameter in Jointly Sparse Global SIMPLS-R are selected by 2-fold cross validation to minimize the sum of the MSE of the responses. Since each subject has multiple samples, we perform the cross validation by splitting by subjects, i.e., no samples from the same subject will appear in both training and tuning sets. We restrict the  responses to the first 3 symptoms, which are the upper respiratory symptoms, and the results are shown in Table \ref{table:PLSsymptom3detail}. In most of the cases, the proposed Jointly Sparse Global SIMPLS-R method outperforms the standard PLS-R and $\ell_1$ penalized PLS-R in terms of prediction MSE, number of components, number of genes. As can be seen in Table \ref{table:PLSsymptom3detail}, the number of selected variables decreases significantly by applying the $\ell_1/\ell_2$ mixed norm sparsity penalty to the PLS-R objective function. Thus the proposed PLS-R method is able to construct a more parsimonious predictor relative to the other PLS-R methods having similar accuracy.

The PLS-R method can also be viewed as an exploratory data analysis tool for constructing low dimensional descriptors of the independent variables and response variables. Specifically, the general underlying matrix factorization 
model $X=TP'+E$ and $Y=TQ'+F$, with latent component $T=XW$, provides a factor analysis model for the independent and response variables $X$ and $Y$. $T$, $P$ and $Q$  can be interpreted in a similar manner as the singular vectors of  PCA. However, different from PCA that does not account for the response variables, $T$, $P$ and $Q$ contain information about both the independent variables and the response variables. The factor analysis interpretation of the underlying PLS model is that $T$ is a latent score matrix and $P$, $Q$ are latent factor loading matrices that associate $T$ with the independent variables and the response variables via the approximate matrix factorizations $X\approx TP'$ and $Y\approx TQ'$, respectively.  The correlations between the latent component $T$ and the sum of the 3 upper respiratory symptoms are reported in Table \ref{table:mf}, which also shows results for classic matrix factorization methods including non-negative matrix factorization (NMF) \cite{Paatero1994} and Bayesian linear unmixing (BLU) \cite{Dobigeon2009,Bazot2013}, previously applied to this dataset,  for comparison.  Notice the sparse PLS-R methods achieve higher correlation, as expected.  Remarkably, the proposed Jointly Sparse global SIMPLS-R achieves this higher degree of correlation with many fewer components and variables than the NMF and BLU methods. This experiment demonstrates that Jointly Sparse Global SIMPLS-R can be used as a factor analysis method to find the hidden molecular factors that best relate to the response.

\begin{table}
\caption{Performance comparison table for the predictive health study. We apply the Jointly Sparse Global SIMPLS Regression ($\ell_1/\ell_2$ SPLS) for sparse prediction of disease symptoms from gene expression. The performance is compared with standard PLS and $\ell_1$ sparse PLS.  $\ell_1/\ell_2$ SPLS achieves lower MSE with significantly fewer  variables in most of the studies.  \label{table:PLSsymptom3detail}}
\tiny
\begin{center}
    \begin{tabular}{  *{6}{c} }    \\\\
    \hline \\
&  &   &   &  \multicolumn{2}{c}{p-values of one sided paired t-test} \\
\cline{5-6}\\
 & 1. PLS  & 2. $\ell_1$ SPLS  & 3. $\ell_1/\ell_2$ SPLS  & (3,1) & (3,2)\\  \\ \hline
\\{\bf H1N1}\\
number of comp.  &  2.8 &  2.3 & 2.4 & 0.1163 & 0.3322\\ 
number of genes  & 12023 & 3624.1 &  3575.8 & $1.4451\times10^{-10}$ & 0.4842\\ 
Overall MSE &  0.599 & 0.603 & 0.591 & 0.2890 &  0.1094\\ 
Runny nose MSE & 0.167 & 0.165 & 0.167\\
Stuffy nose MSE & 0.281 & 0.282 & 0.269\\
Sneezing MSE & 0.151 & 0.157 & 0.155\\\hline 
\\{\bf H3N2}\\
number of comp.  &  3.2 & 2.5 &  1.9 & 0.0030 &  0.0863\\ 
number of genes  & 12023 & 3944.5 & 1721.5 & $1.7547\times10^{-9}$& 0.0601\\ 
Overall MSE &  0.623 &  0.622 & 0.609 &  0.3073 &  0.2530\\ 
Runny nose MSE & 0.186 & 0.174 & 0.173\\
Stuffy nose MSE & 0.277 & 0.284 & 0.272\\
Sneezing MSE & 0.160 & 0.164 & 0.165\\\hline 
\\{\bf HRV}\\
number of comp.  &  2.8 & 2.3 &  2.2 &  0.0484 &  0.3773\\ 
number of genes  & 12023 & 2193.2 & 1779.1 & $5.5038\times10^{-13}$ &  0.3522\\ 
Overall MSE &  0.628 &  0.607 &  0.603 &  0.2020 &  0.4490\\ 
Runny nose MSE & 0.243 & 0.226 & 0.232\\
Stuffy nose MSE & 0.324 & 0.323 & 0.314\\
Sneezing MSE & 0.062 & 0.058 & 0.057\\\hline 
\\{\bf RSV}\\
number of comp.  & 3.2 & 2.3 & 2.4 & 0.0198 & 0.4103\\ 
number of genes  & 12023 & 2445.4 & 3889.8 & $1.1584\times10^{-9}$ & 0.1472\\ 
Overall MSE &  0.866 & 0.920 & 0.855 & 0.3318 & 0.0567\\ 
Runny nose MSE & 0.312 & 0.327 & 0.312\\
Stuffy nose MSE & 0.412 & 0.448 & 0.397\\
Sneezing MSE & 0.143 & 0.145 & 0.145\\\hline 
    \end{tabular}
    \end{center}

\end{table}

\begin{table}
\tiny
\caption{Matrix factorization. We use the cross-validated parameters reported in Table \ref{table:PLSsymptom3detail} for standard PLS-R and $\ell_1$ penalized PLS-R, and the Jointly Sparse Global SIMPLS-R to decide the number of components, and search over the grid $\{1,2,...,10\}$ to find the number of factors that achieves the highest correlation for NMF and BLU. The correlation between each factor and the sum of responses is listed for each method. The first 3 methods, PLS, $\ell_1$ SPLS, and $\ell_1/\ell_2 $ SPLS are supervised matrix factorizations, where as NMF and BLU are unsupervised. The unsupervised methods require many more factors to achieve comparable correlation.  \label{table:mf}}
\begin{center}  \begin{tabular}{  *{11}{c} }
    \hline
 & \multicolumn{10}{c}{correlation of each factor with the sum of upper respiratory symptoms}\\ 
factor & 1 & 2 & 3 & 4 & 5 & 6 & 7 & 8 & 9 \\\hline
\\{\bf H1N1}\\
PLS &  0.47 &  0.38  & 0.38 \\
$\ell_1$ SPLS &  0.52  & 0.37\\
$\ell_1/\ell_2$ SPLS &   0.57 & 0.35 \\
NMF & 0.32 &   0.45 & 0.04\\
BLU &   0.27 & 0.19 & 0.02 & 0.51 &  0.16  &  0.06 \\\hline
\\{\bf H3N2}\\
PLS & 0.67 &   0.42  &  0.33\\
$\ell_1$ SPLS &   0.73  &  0.33 & 0.33\\
$\ell_1/L_2$ SPLS & 0.71  &  0.33 \\
NMF &  0.62 &  0.70 &  0.10\\
BLU &   0.54 & 0.73 & 0.26 & 0.33 & 0.01 & 0.02 &  0.28 & 0.05 & 0.00 \\\hline
\\{\bf HRV}\\
PLS &   0.45 & 0.43 & 0.35\\
$\ell_1$ SPLS &   0.52 & 0.38\\
$\ell_1/\ell_2$ SPLS &  0.53 & 0.42 \\
NMF &   0.02  &  0.22  &  0.19\\
BLU &  0.11  &  0.04 &   0.18  &  0.33 &   0.01 &   0.02 &  0.28 &   0.05  &  0.00 \\\hline
\\{\bf RSV}\\
PLS & 0.66 &   0.35 &   0.35\\
$\ell_1$ SPLS &  0.70   & 0.34\\
$\ell_1/\ell_2$ SPLS & 0.69 &   0.39\\
NMF & 0.41 &  0.13   & 0.16 &  0.01  &  0.31  & 0.02 &  0.11  & 0.11 &  0.01\\
BLU &  0.01 &   0.02 &  0.23  & 0.20 &   0.03 &   0.68 &  0.12 \\\hline
\end{tabular}\end{center}
\begin{tabular}{{c}}
\\
   \end{tabular}
   \end{table}


\section{Conclusion}
\label{sec:conclusion}
 The formulation of the global SIMPLS objective function with an added group sparsity penalty greatly reduces the number of variables used to predict the response. This suggests that when multiple components are desired, the variable selection technique should take into account the sparsity structure for the same variables among all the components. Our proposed Jointly Sparse Global SIMPLS Regression algorithm is able to achieve as good or better performance with fewer predictor variables and fewer components as compared to competing methods.  It is thus useful for performing dimension reduction and variable selection simultaneously in applications with large dimensional data but comparatively  few samples ($n<p$). 

The Jointly Sparse Global SIMPLS Regression objective function is minimized using augmented Lagrangian techniques and, in particular, the ADMM algorithm. The ADMM algorithm splits the optimization into an eigen-decomposition problem and a soft-thresholding that enforces sparsity constraints. The general framework is extendable to more complicated regularization and can thus be tailored for other PLS-type applications, e.g., positivity constraints or smoothness penalties.  For example, in the chemometric application, the data is smooth over the wavelengths and we can apply wavelet shrinkage on the data or include a total variation regularization to encourage smoothness. The sparsity constraints can be imposed on the wavelet coefficients if wavelet shrinkage is applied, or together with total variation regularization. The equivalence of soft wavelet shrinkage and total variation regularization was discussed in \citep{steidl2004equivalence}. One can also consider imposing sparsity structures on the weights corresponding to the same components, adding $\ell_1$ penalty within the groups, or total variation regularization, depending on the applications. The decoupling property of the ADMM algorithm allows one to extend the Jointly Sparse Global SIMPLS Regression to these various regularizations.


\end{document}